\def\isarxiv{1}

\def\paperTitle{Lazy Kronecker Product}

\def\paperAuthor{
Zhao Song\thanks{\texttt{magic.linuxkde@gmail.com}.} 
}

\ifdefined\isarxiv
\documentclass[11pt]{article}
\usepackage[numbers]{natbib}
\else
\documentclass{article} 
\usepackage{iclr2026_conference,times}

\fi

\ifdefined\isarxiv

\usepackage{amsmath}
\usepackage{amsthm}
\usepackage{amssymb}
\usepackage{algorithm}
\usepackage{subfig}
\usepackage{algpseudocode}
\usepackage{graphicx}
\usepackage{grffile}
\usepackage{wrapfig,epsfig}
\usepackage{url}
\usepackage{xcolor}
\usepackage{epstopdf}

\usepackage{bbm}
\usepackage{dsfont}

\else 

\usepackage{hyperref}
\usepackage{url}

\fi

\ifdefined\isarxiv

\else

\usepackage{amsmath}
\usepackage{amsthm}
\usepackage{amssymb}
\usepackage{algorithm}
\usepackage{subfig}
\usepackage{algpseudocode}
\usepackage{graphicx}
\usepackage{grffile}
\usepackage{wrapfig,epsfig}
\usepackage{url}
\usepackage{xcolor}
\usepackage{epstopdf}

\usepackage{bbm}
\usepackage{dsfont}

\fi
 
\allowdisplaybreaks

\ifdefined\isarxiv

\usepackage{tikz}
\usepackage{hyperref}  
\hypersetup{colorlinks=true,citecolor=blue,linkcolor=blue} 
\usetikzlibrary{arrows}
\usepackage[margin=1in]{geometry}

\else
\fi
 
\graphicspath{{./figs/}}

\theoremstyle{plain}
\newtheorem{theorem}{Theorem}[section]

\newtheorem{definition}[theorem]{Definition}

\newtheorem{conjecture}[theorem]{Conjecture}

\ifdefined\isarxiv

\else
\renewcommand\cite\citep
\fi

\newcommand{\R}{\mathbb{R}}

\DeclareMathOperator{\diag}{diag}

\begin{document}

\ifdefined\isarxiv

\date{}
\title{\paperTitle}
\author{\paperAuthor}

\else

\title{\paperTitle}

\author{Antiquus S.~Hippocampus, Natalia Cerebro \& Amelie P. Amygdale \thanks{ Use footnote for providing further information
about author (webpage, alternative address)---\emph{not} for acknowledging
funding agencies.  Funding acknowledgements go at the end of the paper.} \\
Department of Computer Science\\
Cranberry-Lemon University\\
Pittsburgh, PA 15213, USA \\
\texttt{\{hippo,brain,jen\}@cs.cranberry-lemon.edu} \\
\And
Ji Q. Ren \& Yevgeny LeNet \\
Department of Computational Neuroscience \\
University of the Witwatersrand \\
Joburg, South Africa \\
\texttt{\{robot,net\}@wits.ac.za} \\
\AND
Coauthor \\
Affiliation \\
Address \\
\texttt{email}
}

%

\newcommand{\fix}{\marginpar{FIX}}
\newcommand{\new}{\marginpar{NEW}}

\maketitle

\fi

\ifdefined\isarxiv
  \maketitle
  \begin{abstract}
    In this paper, we show how to generalize the lazy update regime from dynamic matrix product [Cohen, Lee, Song STOC 2019, JACM 2021] \cite{cls21} to dynamic kronecker product. We provide an algorithm that uses $n^{\omega( \lceil k/2 \rceil, \lfloor k/2 \rfloor, a )-a}$ amortized update time and $ n^{\omega( \lceil(k-s)/2 \rceil, \lfloor (k-s)/2 \rfloor,a )}$ worst case query time for dynamic kronecker product problem. Unless tensor MV conjecture is false, there is no algorithm that can use both $n^{\omega( \lceil k/2 \rceil, \lfloor k/2 \rfloor, a )-a-\Omega(1)}$ amortized update time, and $ n^{\omega( \lceil(k-s)/2 \rceil, \lfloor (k-s)/2 \rfloor,a )-\Omega(1)}$ worst case query time.

  \end{abstract}


\else

\begin{abstract}
In this paper, we show how to generalize the lazy update regime from dynamic matrix product [Cohen, Lee, Song STOC 2019, JACM 2021] \cite{cls21} to dynamic kronecker product. We provide an algorithm that uses $n^{\omega( \lceil k/2 \rceil, \lfloor k/2 \rfloor, a )-a}$ amortized update time and $ n^{\omega( \lceil(k-s)/2 \rceil, \lfloor (k-s)/2 \rfloor,a )}$ worst case query time for dynamic kronecker product problem. Unless tensor MV conjecture is false, there is no algorithm that can use both $n^{\omega( \lceil k/2 \rceil, \lfloor k/2 \rfloor, a )-a-\Omega(1)}$ amortized update time, and $ n^{\omega( \lceil(k-s)/2 \rceil, \lfloor (k-s)/2 \rfloor,a )-\Omega(1)}$ worst case query time.
\end{abstract}

\fi



\section{Introduction}

We define a dynamic kronecker product problem such that, at each iteration, we receive a rank-$1$ update and a multi-dimensional query set. This can be viewed as a generalization of the matrix version presented in \cite{cls21}. Consider a $k$-th order tensor $A \in \R^{n \times \cdots \times n}$ and a positive integer $s$, an index set $\{ \ell_1, \cdots, \ell_s \} \subset [k]$ and a multi-set $i = \{ i_1, \cdots, i_s \}$ where each entry is chosen from $[n]$. Here $A_{\cdots, i_1, \cdots, i_s, \cdots}$ denotes a subtensor of $A$ that selects $i_1$ in $\ell_1$ direction, $i_2$ in $\ell_2$ direction and $i_s$ in $\ell_s$ direction.

\begin{definition}[Dynamic kronecker product]\label{def:dynamic_tensor}
Suppose we are given vectors $\{ u_{1}(1), \cdots, u_{1}(T) \} \subset \R^n$, $\{ u_{2}(1), \cdots, u_{2}(T) \} \subset \R^n, \cdots, $ $\{ u_{k}(1), \cdots, u_{k}(T) \} \subset \R^n$. Let $s$ denote a positive integer. During each iteration $t$, we will receive $(u_1(t) , u_2(t) ,\cdots, u_k(t) ,  \ell(t), i(t) )$.  Here $\ell= \{\ell_1, \cdots, \ell_s\}$ is an index set which is a subset of $[k]$. And $i = \{ i_1, i_2, \cdots, i_s \}$ is a multi-index set where each element is chosen from $[n]$. The goal is to output the answer
$(\sum_{i=1}^t \otimes_{j=1}^k u_{j}(i) )_{\cdots,i_1,\cdots,i_s, \cdots} $
for each iteration $t \in [T]$. We remark that the output is a sub-tensor.

\end{definition}

We state our main result as follows.
\begin{theorem}[Our result, algorithm]
Let $k \geq s \geq 1$ be positive integers. 
For any $a>0$, there is an algorithm that takes 
$
O( n^{ \omega( \lceil k/2 \rceil, \lfloor k/2 \rfloor,a ) - a} )$ amortized update time and $ O( n^{\omega( \lceil ( k - s ) / 2 \rceil, \lfloor ( k - s ) / 2 \rfloor, a )} )$ worst case query time.
\end{theorem}
\begin{proof}
Let $K:= n^a$. Without loss of generality, let us assume $K$ is a positive integer. We maintain a tensor $A$ and low-rank components $U_1,U_2, \cdots, U_k \in \R^{n \times K}$. After every $K$ iterations, we will update $A$ by computing $A \gets A +U_1 \otimes U_2 \otimes \cdots \otimes U_k$. In each iteration, the query time has two parts: the full part and the low-rank part. (1) Full part. In this part, we compute $A_{\cdots,i_1,\cdots,i_s,\cdots}$, which takes $O(n^{k-s})$ time.  (2) Low-rank part. In this part, we consider the time to compute $ (U_1 \otimes U_2 \otimes \cdots \otimes U_k)_{\cdots, i_1, \cdots, i_s, \cdots} = U_1 \otimes \cdots \otimes v_{l_1} \otimes \cdots \otimes v_{l_s} \otimes \cdots \otimes U_k$. Here we replace $s$ tensors $U_{\ell_1}, \cdots, U_{\ell_s}$ by $s$ vectors $v_{\ell_1}, \cdots, v_{\ell_s}$.  Let $w'$ denote a vector that is constructed by $\odot_{i=1}^s v_{\ell_i}$. The operation $\odot$ denotes the Hadamard product. Let $W'$ denote one of the un-replaced matrices. We construct $V_1$ by computing $\diag(w') W'$. Let $V_2, \cdots, V_{k-s}$ denote the remaining un-replaced matrices.  We compute $B = \oslash_{j=1}^{\lceil (k-s) / 2 \rceil} V_j $.\footnote{Given matrices $A_1 \in \R^{n \times d}, A_2 \in \R^{n \times d}, \cdots A_k \in \R^{n \times d}$, we define the $n^k \times d$ matrix $B := \oslash_{j=1}^k A_j$ as the $(p_1,\cdots,p_k) , q$ entry is $\prod_{j=1}^k (A_{j} )_{p_j,q}$.} Similarly, let us compute $C = \oslash_{j=\lceil (k-s)/2 \rceil+1 }^{ k-s } V_j $. 
Thus, the low-rank part takes $n^{\omega( \lceil(k-s)/2 \rceil, \lfloor (k-s)/2 \rfloor,a )}$ time in total. After every $K$ iterations, we need to update $A$. This update takes $n^{\omega( \lceil k/2 \rceil, \lfloor k/2 \rfloor, a )}$ time. To see this, we decompose the computation of $\otimes_{j=1}^{k} U_j $ for $U_1, \cdots, U_k \in \R^{n \times K}$ into two steps. First, we construct an intermediate matrix $B = \oslash_{j=1}^{\lceil k/2 \rceil} U_j$ and $C = \oslash_{j= \lceil k /2 \rceil+1 }^{k} U_j$.  Next, we compute $B \cdot C^\top$, which takes $n^{\omega( \lceil k/2 \rceil, \lfloor k/2 \rfloor, a )}$ time (the $n^{\lceil k/2 \rceil} \times n^{\lfloor k/2 \rfloor}$ matrix has a natural one-to-one mapping to a $k$-th order tensor $n \times n \times \cdots \times n$ tensor). Therefore, the time for updating $A$ is dominated by the second step, so the total time per update is $n^{\omega( \lceil k/2 \rceil, \lfloor k/2 \rfloor, a )}$. Since the update is performed once every $K$ iterations, the amortized update time is $n^{\omega( \lceil k/2 \rceil, \lfloor k/2 \rfloor, a )-a}$.
Thus, it takes $
O( n^{ \omega( \lceil k/2 \rceil, \lfloor k/2 \rfloor,a ) - a} )$ amortized update time and $ O( n^{\omega( \lceil ( k - s ) / 2 \rceil, \lfloor ( k - s ) / 2 \rfloor, a )} )$ worst case query time.
\end{proof}
We remark that for $k=2$ and $s=1$, the running time degenerates to update time of Lemma 5.4 and query time part of Lemma 4.5 in \cite{cls21}.

\cite{bns19} proposed hinted Mv conjecture, and \cite{s26} generalized to tensor.
\begin{definition}[A Tensor Version of Hinted Mv, \cite{s26}]\label{def:tensor_hinted_mv:otimes}
Working over a boolean semi-ring for a fixed parameter $\tau > 0$, we define the Tensor Hinted Mv problem as a three-phase process: Phase 1: Receive $k$ matrices $V_1, V_2, \dots, V_k$, each of size $n \times d$. Phase 2: Receive an order-$k$ diagonal tensor $P$ of dimensions $d \times d \times \dots \times d$, containing at most $n^\tau$ non-zero entries. Phase 3: Receive a set of target modes $\{ \ell_1, \dots, \ell_s \} \subseteq [k]$ and a corresponding multi-set of indices $\{ i_1, \dots, i_s \}$ drawn from $[n]$. The goal is to output the sub-tensor $[P(V_1, V_2, \dots, V_k)]_{\dots, i_1, \dots, i_s, \dots}$, where each index $i_t$ restricts the tensor along the corresponding $\ell_t$-th direction, for all $t\in [k]$. 
\end{definition}

\begin{conjecture}[\cite{s26}]\label{con:main}
For any algorithm solving the Tensor Hinted Mv problem (Definition~\ref{def:tensor_hinted_mv:otimes}) using polynomial preprocessing in Phase 1, at least one of the following lower bounds must hold: Phase 2 requires $\Omega(n^{\omega( \lceil k/2 \rceil , \lfloor k/2 \rfloor, \tau)-\delta} )$ time, Phase 3 requires $\Omega(n^{\omega( \lceil (k-s)/2 \rceil, \lfloor (k-s)/2 \rfloor , \tau)-\delta} )$ time 
for every $\delta > 0$.
\end{conjecture}

\begin{theorem}[Our result, hardness]
Unless tensor MV conjecture (Conjecture~\ref{con:main}) is false,  there is no algorithm that can use both $n^{\omega( \lceil k/2 \rceil, \lfloor k/2 \rfloor, a )-a-\delta}$ amortized update time, and $ n^{\omega( \lceil(k-s)/2 \rceil, \lfloor (k-s)/2 \rfloor,a )-\delta}$ worst query time for some constant $\delta >0$.
\end{theorem}
\begin{proof}
We will prove this by reduction. Assume, for the sake of contradiction, that there exists a dynamic algorithm $\mathcal{A}$ for the dynamic tensor multiplication problem (Definition~\ref{def:dynamic_tensor}) that simultaneously achieves: amortized update time: $O(n^{\omega(\lceil k/2 \rceil, \lfloor k/2 \rfloor, a) - a - \delta})$, and worst case query time: $O(n^{\omega(\lceil (k-s)/2 \rceil, \lfloor (k-s)/2 \rfloor, a) - \delta})$. We will show how to use algorithm $\mathcal{A}$ to solve the Tensor Hinted Mv problem (Definition~\ref{def:tensor_hinted_mv:otimes}) faster than the lower bounds specified in Conjecture~\ref{con:main}, thus breaking the conjecture. To align the parameters, we set $a = \tau$. In Phase 1 of the Tensor Hinted Mv problem, we are given $k$ matrices $V_1, V_2, \dots, V_k$, each of size $n \times d$. We initialize an empty dynamic tensor structure using algorithm $\mathcal{A}$. This phase acts as our polynomial preprocessing. In Phase 2, we receive a diagonal tensor $P$ of size $d \times d \times \dots \times d$ with at most $n^\tau$ non-zero entries. Let the number of non-zero entries be $m \le n^\tau$. Each non-zero entry $P_{j,j,\dots,j}$ corresponds to a rank-1 update. For each of the $m$ non-zero entries, we issue an update to our dynamic algorithm $\mathcal{A}$. The update vectors are simply the $j$-th columns of our given matrices: $u_1(t) = (V_1)_{*,j}$,  $u_2(t) = (V_2)_{*,j}$, $ \dots$, $u_k(t) = (V_k)_{*,j}$.
Since there are at most $n^\tau$ such updates, and the amortized update time of $\mathcal{A}$ is $O(n^{\omega(\lceil k/2 \rceil, \lfloor k/2 \rfloor, \tau) - \tau - \delta})$, the total time required to process all updates is bounded by:
$m \times O(n^{\omega(\lceil k/2 \rceil, \lfloor k/2 \rfloor, \tau) - \tau - \delta}) \le n^\tau \times O(n^{\omega(\lceil k/2 \rceil, \lfloor k/2 \rfloor, \tau) - \tau - \delta}) = O(n^{\omega(\lceil k/2 \rceil, \lfloor k/2 \rfloor, \tau) - \delta})$. In Phase 3, we receive an index set $\{\ell_1, \dots, \ell_s\} \subseteq [k]$ and a multi-set of indices $\{i_1, \dots, i_s\}$. We need to evaluate the sub-tensor $[P(V_1, \dots, V_k)]_{\dots, i_1, \dots, i_s, \dots}$.
This is exactly the query operation defined for our dynamic tensor multiplication problem (Definition~\ref{def:dynamic_tensor}). We query algorithm $\mathcal{A}$ with these indices. By our assumption, the worst-case query time for algorithm $\mathcal{A}$ is: $O(n^{\omega(\lceil (k-s)/2 \rceil, \lfloor (k-s)/2 \rfloor, \tau) - \delta})$. We have constructed a solver for the Tensor Hinted Mv problem with the following runtimes: Phase 2 Time: $O(n^{\omega(\lceil k/2 \rceil, \lfloor k/2 \rfloor, \tau) - \delta})$. Phase 3 Time: $O(n^{\omega(\lceil (k-s)/2 \rceil, \lfloor (k-s)/2 \rfloor, \tau) - \delta})$. However, Conjecture~\ref{con:main} states that for \emph{any} algorithm, at least one of the following must be true for every $\delta > 0$:  Phase 2 requires $\Omega(n^{\omega(\lceil k/2 \rceil, \lfloor k/2 \rfloor, \tau) - \delta})$. Phase 3 requires $\Omega(n^{\omega(\lceil (k-s)/2 \rceil, \lfloor (k-s)/2 \rfloor, \tau) - \delta})$. Since our algorithm $\mathcal{A}$ strictly beats \emph{both} bounds simultaneously (by a factor of $n^\delta$), we have reached a contradiction. Therefore, such an algorithm $\mathcal{A}$ cannot exist unless the Tensor Hinted Mv conjecture is false.
\end{proof}
We remark that the above proof is similar to proof of Theorem 5.3 in \cite{bns19}.

\ifdefined\isarxiv
\bibliographystyle{alpha}
\bibliography{ref}
\else
\bibliographystyle{iclr2026_conference}
\bibliography{ref}
\fi

\end{document}